\documentclass[12pt,a4paper]{article}
\usepackage{amssymb,amsmath,amsthm}
\usepackage{graphicx}
\usepackage{algorithmicx}
\usepackage{algpseudocode}
\usepackage{lscape}
\usepackage{longtable}

\usepackage{algorithm}
\usepackage{url}
\usepackage{algcompatible}
 \usepackage{setspace}

\begin{document}


\title{ An enhanced pinwheel algorithm for the bamboo garden trimming problem}
\author{Federico Della Croce\footnote{DIGEP, Politecnico di Torino, Italy, \texttt{federico.dellacroce@polito.it}}
\hspace*{0.02cm} 
\footnote{CNR, IEIIT, Torino, Italy}
}

\date{}
\maketitle

\newtheorem{rem}{Remark}
\newtheorem{theorem}{Theorem}
\newtheorem{prop}{Proposition}
\newtheorem{coro}{Corollary}
\newtheorem{lemma}{Lemma}

\begin{abstract}
In the Bamboo Garden Trimming Problem (BGT), there is a garden  populated by $n$ bamboos $b_1, b_2, \cdots, b_n$ with daily growth rates $h(1)\geq h(2) \geq \cdots \geq h(n)$. We assume that the initial heights of bamboos are zero. A gardener is in charge of the bamboos  and trims them to height zero according to some schedule. The objective is to design a perpetual schedule of trimming so as to maintain the height of the bamboo garden as low as possible. 
We consider the so-called 
discrete BGT variant, where the gardener is allowed to trim only one bamboo at the end of each day.
For discrete BGT, the current state-of-the-art approximation algorithm 
exploits the relationship between BGT and the classical Pinwheel scheduling problem and
provides a solution that guarantees a $2$-approximation ratio.
We 
propose an alternative Pinwheel scheduling algorithm 
with approximation ratio converging to $\frac{12}{7}$ when $\sum_{j=1}^n h(j) > > h(1)$.
Also, we show that the approximation ratio of the proposed algorithm never exceeds $\frac{32000}{16947} \approx 1.888$. 
\end{abstract}

{\bf Keywords:}
Bamboo trimming, Pinwheel scheduling, Approximation.

\section{Introduction}
We consider the so-called Bamboo Garden Trimming (BGT) Problem
\cite{DDN19,GKLLMR17}.
A garden G is populated by n bamboos $b_1,..., b_n$ each with its own daily growth rate.
It is assumed that
the initial bamboos heights are zero. 
Every day all bamboos except the ones that are cut 
grow the related extra heights.
The goal is to design a perpetual
schedule of cuts to maintain the elevation of the bamboo garden, that is the maximum height ever reached by any of the bamboos, as low as
possible. The problem takes its origins
from perpetual testing of virtual machines in cloud systems \cite{AKG15}.
Let denote by $H^*$ the optimal solution value and, for any algorithm $A$, let denote by $H^A$ the corresponding solution value.
Approximation results are given both in \cite{DDN19} and \cite{GKLLMR17} with one bamboo cut per day.
In \cite{GKLLMR17}, it is shown that a simple strategy denoted $ReduceMax$ that imposes to cut always the bamboo reaching the maximum height in each day has approximation ratio 
$\frac{H^{ReduceMax}}{H^*} \leq \log_2 n$.
Then, the so-called pinwheel algorithm (see \cite{HMRTV89}), here denoted $PW$,
is shown to have approximation ratio $\frac{H^{PW}}{H^*} \leq 2$.
In \cite{DDN19} it is conjectured that
$\frac{H^{ReduceMax}}{H^*} \leq 2$ and it is shown experimentally that such limit is never exceeded.
This work is connected to several pinwheel problems such as the periodic Pinwheel problem \cite{HRTV92,LL97} and the Pinwheel scheduling problem \cite{RR97}.
As mentioned in \cite{DDN19}, the garden with
$n$ bamboos is an analogue of a system of $n$ machines which have to
be attended (e.g., serviced) with different frequencies. \\

The $BGT$ problem can be expressed as follows.
There are $n$ bamboos and for each bamboo $b_j$
the related growth is denoted by $h(j)$.
W.l.o.g. we assume $ h(1) =2^0=1 \geq ... \geq h(n)$.
In \cite{GKLLMR17}, it is shown that a lower bound $LB$ on the maximum height of a bamboo is $\sum_{j=1}^n h(j)$.
Actually, it is straightforward to slightly extend this result as 
$LB = \max \{2h(1);\sum_{j=1}^n h(j)\}$ except for the trivial case with $n=1$.

\bigskip
Here, we look for a periodic trimming of the bamboos. More specifically, we first recall the main features of the 
pinwheel algorithm where each bamboo is assigned to a partition so that, typically, bamboo $b_1$ is assigned to partition $P_1$, bamboos $b_2,...,b_j$ are assigned to partition $P_2$, bamboos $b_{j+1},...,b_k$ to partition $P_3$ and so on up to bamboo $b_n$. If the number of partitions that is reached 
is equal to $\alpha$, then bamboo $b_1$ is repeatedly cut every $\alpha$ partitions and its maximum height will be 
$\alpha \times h(1)$.
If each bamboo $b_j$ has $h(j) = \frac{1}{2^{k}}$ for different integer values of $k\geq 0$, then it is always possible to assign bamboos to all partitions in such a way that
the sum of heights of the bamboos in each partition is equal to $h(1)$ except possibly the last partition that has sum of heights $\leq h(1)$.

Consider the example in Table \ref{tab1} where we have $\sum_{j=1}^n h(j)=135/16 \approx 8.4$.
\begin{table*}
\scriptsize
\begin{center}
\hspace*{-2cm}
\begin{tabular}{||c|c|c|c|c|c|c|c|c|c|c|c|c|c|c|c|c||} \hline 
$b_j$ & $b_1$ & $b_2$ & $b_3$ & $b_4$ & $b_5$ & $b_6$ & $b_7$ & $b_8$ & $b_9$ & $b_{10}$ & $b_{11}$ & $b_{12}$ & $b_{13}$ & $b_{14}$ & $b_{15}$ & $b_{16}$ \\ \hline
$h(j)$  &   $1$ & $1$ & $1$ & $1$ & $\frac{1}{2}$ & $\frac{1}{2}$ & $\frac{1}{2}$ & $\frac{1}{2}$ & $\frac{1}{2}$ &  $\frac{1}{2}$ &  $\frac{1}{2}$ & $\frac{1}{4}$ &   $\frac{1}{4}$  &  $\frac{1}{4}$ &  $\frac{1}{8}$ & $\frac{1}{16}$ \\ \hline
$P_j$  &   
 \multicolumn{1}{|c|}{$P_1:[b_1]$} & 
 \multicolumn{1}{|c|}{$P_2:[b_2]$} &
 \multicolumn{1}{|c|}{$P_3:[b_3]$} &  
 \multicolumn{1}{|c|}{$P_4:[b_4]$} &  
 \multicolumn{2}{|c|}{$P_5:[b_5,b_6]$} &  
 \multicolumn{2}{|c|}{$P_6:[b_7,b_8]$} &  
 \multicolumn{2}{|c|}{$P_7:[b_9,b_{10}]$} &  
 \multicolumn{3}{|c|}{$P_8:[b_{11},b_{12},b_{13}]$} &  
 \multicolumn{3}{|c|}{$P_9:[b_{14},b_{15},b_{16}]$}  \\ \hline
$\sum_{i \in P_j} h(i)$ & 
 \multicolumn{1}{|c|}{$1$} & 
 \multicolumn{1}{|c|}{$1$} & 
 \multicolumn{1}{|c|}{$1$} &
 \multicolumn{1}{|c|}{$1$} &  
 \multicolumn{2}{|c|}{$1$} &  
 \multicolumn{2}{|c|}{$1$} &  
 \multicolumn{2}{|c|}{$1$} &  
 \multicolumn{3}{|c|}{$1$} &  
 \multicolumn{3}{|c|}{$\frac{7}{16}$}  \\ \hline
\end{tabular} 
\end{center} 
\caption{A $16$-bamboos instance with all growths $h(j) = 2^k$ for different integer values of $k\leq 0$}
\label{tab1}
\end{table*}
\normalsize
Row 1 lists the bamboos that are present,
row 2 indicates the height $h(j)$ of each bamboo $b_j$, row 3 indicates how bamboos are split into the partitions
and row $4$ indicates the bamboos heights in each partition. In this case, in all partitions, the sum of heights is equal to $1$ except the last one that has height $\frac{7}{16}$.
The bamboos trimming respects the order of the considered partitions so that partitions are repeated one after another
and every time a partition is considered, then  a bamboo assigned to that partition will be trimmed.
Given a partition $P_i$ and a bamboo $b_j$, the frequency of the bamboo $b_j$ to be trimmed in partition $P_i$ is not smaller than the ratio $r(j) = h(1)/h(j)$.
Indeed bamboo $b_j$ is trimmed at most once every $r(j)$ appearances of partition $P_i$.

\noindent
For the considered example, the $r(j$)s are as indicated in Table \ref{tab1a}
(notice that, if  $h(j) = \frac{1}{2^k}$ for some integer $k \geq 0$, then $r(j)=2^k$ always holds) while partitions and selected bamboos are iterated as indicated in Table \ref{tab1b} and so on (notice that for partition $P_9$, as $\sum_{i \in P_9} h(i) =\frac{7}{16}< 1$, bamboos 
$b_{14},b_{15},b_{16}$ can be trimmed more often than necessary).
Then, for the considered example, the maximum height reached (that is the solution cost) is not larger than $\alpha \times r(j) \times  h(j) = \alpha \times h(1) = 9*1=9$. 

\scriptsize
\begin{table*}
\begin{tabular}{||c|c|c|c|c|c|c|c|c|c|c|c|c|c|c|c|c||} \hline 
$b_j$ & $b_1$ & $b_2$ & $b_3$ & $b_4$ & $b_5$ & $b_6$ & $b_7$ & $b_8$ & $b_9$ & $b_{10}$ & $b_{11}$ & $b_{12}$ & $b_{13}$ & $b_{14}$ & $b_{15}$  & $b_{16}$\\ \hline
$r(j)$  &   1 & 1 & 1 & 1 & 2 & 2 & 2 & 2 &  2 &   2 &   2 &   4 &  4 &   4 &   8 & 16\\ \hline
\end{tabular} 
\caption{Relevant $r(j)$s for the example of Table \ref{tab1}}
\label{tab1a}
\end{table*}

\scriptsize
\begin{table*}
\begin{tabular}{||r||r||r||r||r||r||r||r||r||} \hline 
$P_1:[b_1]$  & $P_2:[b_2]$ &  $P_3:[b_3]$ & $P_4:[b_4]$ & $P_5:[b_{5}] $ &
$P_6:[b_7]$  & $P_7:[b_9]$ &  $P_8:[b_{11}]$ & $P_9:[b_{14}]$ \\ \hline
$P_1:[b_1]$  & $P_2:[b_2]$ &  $P_3:[b_3]$ & $P_4:[b_4]$ & $P_5:[b_{6}] $ &
$P_6:[b_8]$  & $P_7:[b_{10}]$ &  $P_8:[b_{12}]$ & $P_9:[b_{15}]$ \\ \hline
$P_1:[b_1]$  & $P_2:[b_2]$ &  $P_3:[b_3]$ & $P_4:[b_4]$ & $P_5:[b_{5}] $ &
$P_6:[b_7]$  & $P_7:[b_9]$ &  $P_8:[b_{11}]$ & $P_9:[b_{14}]$ \\ \hline
$P_1:[b_1]$  & $P_2:[b_2]$ &  $P_3:[b_3]$ & $P_4:[b_4]$ & $P_5:[b_{6}] $ &
$P_6:[b_8]$  & $P_7:[b_{10}]$ &  $P_8:[b_{13}]$ & $P_9:[b_{16}]$ \\ \hline
\end{tabular}
\caption{Relevant partitions for the example of Table \ref{tab1}}
\label{tab1b}
\end{table*}
\normalsize

\bigskip
\noindent
If bamboos growths are general, hence $h(j) \neq \frac{1}{2^k}$ holds, then $PW$ was used in \cite{GKLLMR17} by considering 
modified growths $h'(j)$ determined as follows: 
$h'(j) = \frac{1}{2^{k}}$ such that $ \frac{1}{2^{k+1}} < h(j) \leq \frac{1}{2^{k}}$ .
The rationale of the approach is to determine for every $h(j)$, the closest $ \frac{1}{2^k}$ value ($k$ integer $\geq 0$) not inferior to $h(j)$ and to set $h'(j)=\frac{1}{2^k}$.

To see how algorithm $PW$ works in this case, consider, the example in Table \ref{tab2} with growths $0 < h(j) \leq 1$, 
where the entries have the same meaning of Table \ref{tab1} and an additional row is included presenting the corresponding $h'(j)$
values. In this case we have $\sum_{i=1}^n h(j)= 5.93$. If we substitute each $h(j)$ with the corresponding $h'(j)$,
we can see that partitions and $r(j)$s are identical to those in Table \ref{tab1}
and we have $H^{PW} = \alpha \times h(1) = 9$ achieved by bamboo $b_1$ (and also other bamboos).

\begin{table*}
\scriptsize
\begin{center}
\hspace*{-2.3cm}
\begin{tabular}{||c|c|c|c|c|c|c|c|c|c|c|c|c|c|c|c|c||} \hline 
$b_j$ & $b_1$ & $b_2$ & $b_3$ & $b_4$ & $b_5$ & $b_6$ & $b_7$ & $b_8$ & $b_9$ & $b_{10}$ & $b_{11}$ & $b_{12}$ & $b_{13}$ & $b_{14}$ & $b_{15}$ & $b_{16}$ \\ \hline
$h(j)$  & $1$ &   $0.83$ & $0.6$  & $0.55$ & $0.45$ & $0.4$ & $0.32$ & $0.29$ &  $0.28$  &   $0.27$ &  $0.26$ &  $0.22$ &  $0.16$ & $0.15$  &   $0.1$ & $0.05$\\ \hline
$h'(j)$  &   $1$ & $1$ & $1$ & $1$ & $\frac{1}{2}$ & $\frac{1}{2}$ & $\frac{1}{2}$ & $\frac{1}{2}$ & $\frac{1}{2}$ &  $\frac{1}{2}$ &  $\frac{1}{2}$ & $\frac{1}{4}$ &   $\frac{1}{4}$  &  $\frac{1}{4}$ &  $\frac{1}{8}$ & $\frac{1}{16}$ \\ \hline
$P_j$  &   
 \multicolumn{1}{|c|}{$P_1:[b_1]$} & 
 \multicolumn{1}{|c|}{$P_2:[b_2]$} &
 \multicolumn{1}{|c|}{$P_3:[b_3]$} &  
 \multicolumn{1}{|c|}{$P_4:[b_4]$} &  
 \multicolumn{2}{|c|}{$P_5:[b_5,b_6]$} &  
 \multicolumn{2}{|c|}{$P_6:[b_7,b_8]$} &  
 \multicolumn{2}{|c|}{$P_7:[b_9,b_{10}]$} &  
 \multicolumn{3}{|c|}{$P_8:[b_{11},b_{12},b_{13}]$} &  
 \multicolumn{3}{|c|}{$P_9:[b_{14},b_{15},b_{16}]$}  \\ \hline
$\sum_{i \in P_j} h'(i)$ & 
 \multicolumn{1}{|c|}{$1$} & 
 \multicolumn{1}{|c|}{$1$} & 
 \multicolumn{1}{|c|}{$1$} &
 \multicolumn{1}{|c|}{$1$} &  
 \multicolumn{2}{|c|}{$1$} &  
 \multicolumn{2}{|c|}{$1$} &  
 \multicolumn{2}{|c|}{$1$} &  
 \multicolumn{3}{|c|}{$1$} &  
 \multicolumn{3}{|c|}{$\frac{7}{16}$}  \\ \hline
\end{tabular} 
\end{center} 
\caption{A $16$-bamboos instance with all growths $0 < h(j) \leq 1$}
\label{tab2}
\end{table*}
\normalsize

\bigskip

As mentioned before, it is shown in \cite{GKLLMR17} that $H^{PW} \leq 2 \sum h(j) \leq  2LB \leq 2H^*$ always holds, that is 
$\frac{H^{PW}}{H^*} \leq 2$. The following Lemma shows the asymptotical tightness of that result. 
\begin{lemma}
The bound on the approximation ratio $\rho_1 = \frac{H^{PW}}{H^*} \leq 2$ is asymptotically tight. 
\end{lemma}

\begin{proof}
Consider an instance with
$n+1$ bamboos ($n$ being even), where $h(1)=1$, $h(2)= \cdots = h(n+1)=\frac{1}{2}+\epsilon$, with $\epsilon > 0$ and arbitrarily close to zero.
We have $H^{PW}=n+1$, while $LB=\max \{2h(1),\sum_{j=1}^{n+1}h(j)\}=\frac{n}{2}+1+n\epsilon$ and 
a periodic asymptotically optimal solution exists with $b_1$ assigned to partition $P_1$ and bamboos $b_{2j}, b_{2j+1}$ assigned to partition
$P_j$, $j=1,...n/2$ with value $H^*=\frac{n}{2}+1+(n+2)\epsilon$
which is arbitrarily close to $LB$ for $\epsilon$ small enough.
Also, this induces an approximation ratio asymptotically converging to 
$\frac{n+1}{\frac{n}{2}+1} \approx 2$
for $n \rightarrow \infty$.
\end{proof}
\section{Main result}
In order to compute a periodic solution providing an improved approximation ratio, we consider a different way for 
determining modified growths denoted here as $h''(j)$ and propose a different algorithm $PW''$ applied to these $h''(j)$ values.
For any given  $\frac{1}{2^{k+1}} < h(j) \leq \frac{1}{2^{k}}, k \geq 0$ integer, consider exhaustively splitting the bamboos into the following four subsets
$S_1$, $S_2$, $S_3$,  and $S_4$, compute the related $h''(j)$ values as follows and determine the relevant number of partitions $\pi_i$ ($i =1,\cdots, 4$) induced by these subsets.
\begin{description}
\item[$S_1$]: subset of bamboos $b_j: \frac{2}{3} < h(j) \leq 1$; let $h''(j)=h(1)=1$
and $\pi_1=\sum_{j \in S_1}h''(j)$.
\item[$S_2$]: subset of bamboos $b_j: \frac{1}{2} < h(j) \leq \frac{2}{3}$; for these bamboos, 
consider two options (a) and (b) and relevant modified growths $h''_a(j)$ and $h''_b(j)$: either set 
$h''_a(j)=1$ and $\pi_2 = |S_2|$, or set $h''_b(j)=\frac{1}{2}$ and $\pi_2=\lfloor \frac{|S_2|}{2}\rfloor$.
\item[$S_3$]: subset of bamboos $b_j: h(j) \leq  \frac{1}{2}$ with $\frac{2}{3}(\frac{1}{2^{k}}) < h(j)\leq\frac{1}{2^{k}}$, $k \geq 1$ integer; let $h''(j)=\frac{1}{2^{k}}$ and  $\pi_3=\lfloor \sum_{j \in S_3} h''(j)\rfloor$.
\item[$S_4$]: subset of bamboos $b_j: h(j) \leq  \frac{1}{2}$ with $\frac{1}{2^{k+1}} < h(j)\leq\frac{2}{3}(\frac{1}{2^{k}}) $, $k \geq 1$ integer; let $h''(j)=\frac{2}{3}(\frac{1}{2^{k}})$  and  $\pi_4=\lfloor \sum_{j \in S_4} h''(j)\rfloor$.
\end{description}

Let denote by $sh_i$ the sum of modified growths of subset $S_i$, that is $sh_i = \sum_{j \in S_i}h''(j)$. We remark that in $S_1$, the  $h''(j)$ values are equal to $1$ and thus integer.  Then, in $S_2$, either all  
$h''(j)$ values are equal to $1$ or are equal to $\frac{1}{2}$.
Also, in 
$S_3$ each $h''(j)$ value is multiple of $\frac{1}{2^{k}}$ for some integer $k \geq 1$ (where $k$ may differ from one bamboo to another) and 
in $S_4$   each $h''(j)$ value is multiple of  $\frac{2}{3}(\frac{1}{2^{k}})$ for some integer $k \geq 1$
(also here $k$ may differ from one bamboo to another). Correspondingly, the rationale is that, in subset $S_1$, $sh_1$ is an integer value and induces directly
an integer number of partitions $\pi_1=sh_1$ where to allocate the related bamboos
$b_j$ with unit modified grow $h''(j)=1$. 
Then, consider subset $S_2$. If option (a) is considered 
all $h''(j)$ values are equal to $1$ and thus integer
and $\pi_2 = sh_2 = \sum_{j \in S_2}h''(j)$.
Else option (b) holds. In this case all $h''(j)$ values are equal to $\frac{1}{2}$ and 
all bamboos (except possibly one if the number of bamboos is odd) are allocated in pairs determining $\pi_2=\lfloor \frac{|S_2|}{2}\rfloor$ partitions.
If the number of bamboos in $S_2$  is odd, then 
$sh_2 = \sum_{j \in S_2}h''(j)$ is fractional. Hence,
the last bamboo $b_l$ of this subset with value $h''(l)=\frac{1}{2}$ corresponding to the mantissa of $sh_2$ 
represents a remaining part  of subset $S_2$ to be assigned to some other partition.
Also, 
in subset $S_3$, as all  $h''(j)$ value are multiple of $\frac{1}{2^{k}}$ ($k \geq 1$), it is always possible to determine $\pi_3$
partitions where to allocate all bamboos of this subset except possibly a smaller subset of bamboos $R_{S_3}$
with  $\sum_{b_j \in R_{S_3}}h''(j) < 1$ corresponding to the mantissa of $sh_3$.
As an example, suppose in subset $S_3$ we have six bamboos $b_{j_1},\cdots,b_{j_6}$, with 
$h''(j_1)=h''(j_2)=h''(j_3)=\frac{1}{2}$, 
$h''(j_4)=h''(j_5)=\frac{1}{4}$ and
$h''(j_6)=\frac{1}{8}$ so that $sh_3=\frac{17}{8}$. Then, we have $\pi_3=\lfloor \frac{17}{8} \rfloor =2$ partitions, where to allocate 
bamboos $b_{j_1},\cdots,b_{j_5}$ with bamboos $b_{j_1},b_{j_2}$ in one partition and 
bamboos $b_{j_3},b_{j_4}$ and $b_{j_5}$ in the other partition, while  $b_{j_6} \in R_{S_3}$
represents a remaining part  of subset $S_3$ to be assigned to some other partition.
A similar consideration holds for subset $S_4$.
As all  $h''(j)$ value are multiple of $\frac{2}{3}(\frac{1}{2^{k}})$ ($k \geq 1$), it is always possible to determine $\pi_4$
partitions where to allocate all bamboos of this subset except possibly a smaller subset of bamboos $R_{S_4}$
with  $\sum_{b_j \in R_{S_4}}h''(j) < 1$ corresponding to the mantissa of $sh_4$.
As an example, suppose in subset $S_4$ we have nine bamboos $b_{k_1},\cdots,b_{k_9}$, with 
$h''(k_1)=h''(k_2)=h''(k_3)=h''(k_4)=h''(k_5)=\frac{1}{3}$, 
$h''(k_6)=h''(k_7)=h''(k_8)=\frac{1}{6}$ and 
$h''(k_9)=\frac{1}{12}$ so that $sh_4=\frac{27}{12}$. Then, we have $\pi_4=\lfloor \frac{27}{12} \rfloor =2$ partitions, where to allocate bamboos $b_{k_1},\cdots,b_{k_9}$ with bamboos $b_{k_1},b_{k_2},b_{k_3}$ in one partition and 
bamboos $b_{k_4},b_{k_5},b_{k_6}, b_{k_7}$ in the other partition, while  $b_{k_8},b_{k_9} \in R_{S_4}$ represent a remaining part  of subset $S_4$ to be assigned to some other partition.
We propose the following algorithm denoted $PW''$ that reads in input the number $n$ of bamboos $b_1, \cdots , b_n$ and relevant 
heights $h(1), \cdots , h(n)$, and provides in output the maximum height
attained by the algorithm corresponding to the minimum between
$z(a)$ and $z(b)$ computed as indicated in the algorithm.

\begin{algorithm}[H]
\begin{algorithmic}[PW'']
\STATE \textbf{Input:} BGT problem with $n$ bamboos $b_j$ and relevant growths $h(j)$ with $h(1)=1 \geq \cdots \geq h(n)$. 
\STATE Consider option (a) and compute subsets $S_i$ as indicated above, relevant $h''(j)$ values, $r(j)$ values and $\sum_{i=1}^4\pi_i(a)$; determine subsets $R_{S_i} (i=2,3,4)$ if non-empty and all bamboos $b_k$ $\in R_{S_2},R_{S_3},R_{S_4}$.
\STATE Determine the number of partitions
necessary to allocate all $b_k$ $\in R_{S_3},R_{S_4}$, namely compute $\pi_R (a)= \lceil \sum_{b_j \in 
(R_{S_3} \cup R_{S_4})}{h''(j)} \rceil$.
\STATE Compute the solution value that is $z(a)=\pi_R(a)+\sum_{i=1}^4\pi_i(a)$.
\STATE If $S_2$ is non empty, let denote by $b_{j^*}$ the 
bamboo $b_j$ with value $\frac{1}{2} < h(j) \leq \frac{2}{3}$
with largest value.
Repeat for option (b) the same steps above mentioned for option (a) and compute
 $\pi_R (b)= \lceil \sum_{b_j \in (R_{S_2} \cup R_{S_3} \cup R_{S_4})}{ h''(j) \rceil} $. Correspondingly, compute
 the solution value that is $z(b) = \frac{h(j^*)}{0.5}
(\pi_R(b)+\sum_{i=1}^4\pi_i(b))$, else $z(b)=+\infty$.
\STATE Return $\min\{z(a), z(b)\}$.
\end{algorithmic}
\caption{\textbf{$PW''$}}
\end{algorithm}

Table \ref{tab3} provides the relevant $h''(j)$ values (split into  $h''_a(j)$ and  $h''_b(j)$ for the two options of subset $S_2$), the relevant subsets $S_j, R_{S_j}  (j=1, \cdots,4)$ and the final partitions ($P_j(a)$ (option a)  and $P_j(b)$ (option b)) for the example of Table \ref{tab2} where $R_{S_2}$ and $R_{S_4}$ are empty.

\begin{table*}
\scriptsize
\begin{center}
\hspace*{-3.5cm}
\begin{tabular}{||c|c|c|c|c|c|c|c|c|c|c|c|c|c|c|c|c||} \hline 
$b_j$ & $b_1$ & $b_2$ & $b_3$ & $b_4$ & $b_5$ & $b_6$ & $b_{12}$ & $b_{15}$ & $b_{16}$ & $b_{7}$ & $b_{8}$ & $b_{9}$ & $b_{10}$ & $b_{11}$ & $b_{13}$ & $b_{14}$ \\ \hline
$h(j)$  & $1$ &   $0.83$ & $0.6$  & $0.55$ & $0.45$ & $0.4$ & $0.22$ & $0.1$ &  $0.05$  &   $0.32$ &  $0.29$ &  $0.28$ &  $0.27$ & $0.26$  &   $0.16$ & $0.15$\\ \hline
$h''(j)$ ($h''_a(j)$ for $S_5$)  &   $1$ & $1$ & $1$ & $1$ & $\frac{1}{2}$ & $\frac{1}{2}$ & $\frac{1}{4}$ & $\frac{1}{8}$ & $\frac{1}{16}$ &  $\frac{1}{3}$ &  $\frac{1}{3}$ & $\frac{1}{3}$ &   $\frac{1}{3}$  &  $\frac{1}{3}$ &  $\frac{1}{6}$ & $\frac{1}{6}$ \\ \hline
$h''(j)$ ($h''_b(j)$ for $S_5$)  &   $1$ & $1$ & $\frac{1}{2}$ & $\frac{1}{2}$ & $\frac{1}{2}$ & $\frac{1}{2}$ & $\frac{1}{4}$ & $\frac{1}{8}$ & $\frac{1}{16}$ &  $\frac{1}{3}$ &  $\frac{1}{3}$ & $\frac{1}{3}$ &   $\frac{1}{3}$  &  $\frac{1}{3}$ &  $\frac{1}{6}$ & $\frac{1}{6}$ \\ \hline
$S_j,R_{S_j}$  &   
 \multicolumn{2}{|c|}{$S_1:[b_1,b_2]$} & 
 \multicolumn{2}{|c|}{$S_2:[b_3,b_4]$} &
 \multicolumn{2}{|c|}{$S_3:[b_5,b_6]$} &  
 \multicolumn{3}{|c|}{$R_{S_3}:[b_{12},b_{15},b_{16}]$} &  
 \multicolumn{7}{|c|}{$S_4:[b_7,b_8,b_9,b_{10},b_{11},b_{13},b_{14}]$} \\ \hline
$P_j(a)$ &   
 \multicolumn{1}{|c|}{$P_1:[b_1]$} & 
 \multicolumn{1}{|c|}{$P_2:[b_2]$} &
 \multicolumn{1}{|c|}{$P_3:[b_3]$} &  
 \multicolumn{1}{|c|}{$P_4:[b_4]$} &  
 \multicolumn{2}{|c|}{$P_5:[b_5,b_6]$} &  
 \multicolumn{3}{|c|}{$P_6:[b_{12},b_{15},b_{16}]$} &  
 \multicolumn{3}{|c|}{$P_7:[b_7,b_8,b_{9}]$} &  
 \multicolumn{4}{|c|}{$P_8:[b_{10},b_{11},b_{13},b_{14}]$} 
\\ \hline
$\sum_{i \in P_j(a)} h''(i)$ & 
 \multicolumn{1}{|c|}{$[1]$} & 
 \multicolumn{1}{|c|}{$[1]$} & 
 \multicolumn{1}{|c|}{$[1]$} &
 \multicolumn{1}{|c|}{$[1]$} &  
 \multicolumn{2}{|c|}{$[1]$} &  
 \multicolumn{3}{|c|}{$[\frac{7}{16}]$} &  
 \multicolumn{3}{|c|}{$[1]$} &  
 \multicolumn{4}{|c|}{$[1]$}  \\ \hline
$P_j(b)$ &   
 \multicolumn{1}{|c|}{$P_1:[b_1]$} & 
 \multicolumn{1}{|c|}{$P_2:[b_2]$} &
 \multicolumn{2}{|c|}{$P_3:[b_3,b_4]$} &  
 \multicolumn{2}{|c|}{$P_4:[b_5,b_6]$} &  
 \multicolumn{3}{|c|}{$P_5:[b_{12},b_{15},b_{16}]$} &  
 \multicolumn{3}{|c|}{$P_6:[b_7,b_8,b_{9}]$} &  
 \multicolumn{4}{|c|}{$P_7:[b_{10},b_{11},b_{13},b_{14}]$} 
\\ \hline
$\sum_{i \in P_j(b)} h''(i)$ & 
 \multicolumn{1}{|c|}{$[1]$} & 
 \multicolumn{1}{|c|}{$[1]$} & 
 \multicolumn{2}{|c|}{$[1]$} &
 \multicolumn{2}{|c|}{$[1]$} &  
 \multicolumn{3}{|c|}{$[\frac{7}{16}]$} &  
 \multicolumn{3}{|c|}{$[1]$} &  
 \multicolumn{4}{|c|}{$[1]$}  \\ \hline
\end{tabular} 
\end{center} 
\caption{The modified growths $h''(j)$ and related entries for the $16$-bamboos instance with all growths $0 < h(j) \leq 1$}
\label{tab3}
\end{table*}
\normalsize

Thus, if option (a) is considered, 
partitions and selected bamboos are iterated as indicated in Table \ref{tab2b}
{\scriptsize
\begin{table*}
\begin{tabular}{||r||r||r||r||r||r||r||r||} \hline 
$P_1:[b_1]$  & $P_2:[b_2]$ &  $P_3:[b_3]$ & $P_4:[b_4]$ & $P_5:[b_{5}] $ &
$P_6:[b_{12}]$  & $P_7:[b_7]$ &  $P_8:[b_{10}]$ \\ \hline
$P_1:[b_1]$  & $P_2:[b_2]$ &  $P_3:[b_3]$ & $P_4:[b_4]$ & $P_5:[b_{6}] $ &
$P_6:[b_{15}]$  & $P_7:[b_8]$ &  $P_8:[b_{11}]$ \\ \hline
$P_1:[b_1]$  & $P_2:[b_2]$ &  $P_3:[b_3]$ & $P_4:[b_4]$ & $P_5:[b_{5}] $ &
$P_6:[b_{12}]$  & $P_7:[b_9]$ &  $P_8:[b_{13}]$ \\ \hline
$P_1:[b_1]$  & $P_2:[b_2]$ &  $P_3:[b_3]$ & $P_4:[b_4]$ & $P_5:[b_{6}] $ &
$P_6:[b_{16}]$  & $P_7:[b_7]$ &  $P_8:[b_{10}]$ \\ \hline
$P_1:[b_1]$  & $P_2:[b_2]$ &  $P_3:[b_3]$ & $P_4:[b_4]$ & $P_5:[b_{5}] $ &
$P_6:[b_{12}]$  & $P_7:[b_8]$ &  $P_8:[b_{11}]$ \\ \hline
$P_1:[b_1]$  & $P_2:[b_2]$ &  $P_3:[b_3]$ & $P_4:[b_4]$ & $P_5:[b_{6}] $ &
$P_6:[b_{15}]$  & $P_7:[b_9]$ &  $P_8:[b_{14}]$ \\ \hline
$P_1:[b_1]$  & $P_2:[b_2]$ &  $P_3:[b_3]$ & $P_4:[b_4]$ & $P_5:[b_{5}] $ &
$P_6:[b_{12}]$  & $P_7:[b_7]$ &  $P_8:[b_{10}]$ \\ \hline
$P_1:[b_1]$  & $P_2:[b_2]$ &  $P_3:[b_3]$ & $P_4:[b_4]$ & $P_5:[b_{6}] $ &
$P_6:[b_{16}]$  & $P_7:[b_8]$ &  $P_8:[b_{11}]$ \\ \hline
$P_1:[b_1]$  & $P_2:[b_2]$ &  $P_3:[b_3]$ & $P_4:[b_4]$ & $P_5:[b_{5}] $ &
$P_6:[b_{12}]$  & $P_7:[b_9]$ &  $P_8:[b_{13}]$ \\ \hline
$P_1:[b_1]$  & $P_2:[b_2]$ &  $P_3:[b_3]$ & $P_4:[b_4]$ & $P_5:[b_{6}] $ &
$P_6:[b_{15}]$  & $P_7:[b_7]$ &  $P_8:[b_{10}]$ \\ \hline
$P_1:[b_1]$  & $P_2:[b_2]$ &  $P_3:[b_3]$ & $P_4:[b_4]$ & $P_5:[b_{5}] $ &
$P_6:[b_{12}]$  & $P_7:[b_8]$ &  $P_8:[b_{11}]$ \\ \hline
$P_1:[b_1]$  & $P_2:[b_2]$ &  $P_3:[b_3]$ & $P_4:[b_4]$ & $P_5:[b_{6}] $ &
$P_6:[b_{16}]$  & $P_7:[b_9]$ &  $P_8:[b_{14}]$ \\ \hline
\end{tabular}
\label{tab2b}
\caption{Relevant partitions and bamboos cutting frequencies for the example of Table \ref{tab3}: option a}
\end{table*}}
\bigskip
\normalsize
\noindent
and so on (notice that for partition $P_6(a)$, as $\sum_{i \in P_6(a)} h''(i) =\frac{7}{16}< 1$, bamboos 
$b_{12},b_{15},b_{16}$ can be trimmed more often than necessary).

Alternatively, if option (b) is considered, 
partitions and selected bamboos are iterated as indicated in Table \ref{tab2c}
\bigskip
{\scriptsize
\begin{table*}
\begin{tabular}{||r||r||r||r||r||r||r||r||} \hline 
$P_1:[b_1]$  & $P_2:[b_2]$ &  $P_3:[b_3]$ & $P_4:[b_5]$ & $P_5:[b_{12}] $ &
$P_6:[b_{7}]$  & $P_7:[b_{10}]$  \\ \hline
$P_1:[b_1]$  & $P_2:[b_2]$ &  $P_3:[b_4]$ & $P_4:[b_6]$ & $P_5:[b_{15}] $ &
$P_6:[b_{8}]$  & $P_7:[b_{11}]$  \\ \hline
$P_1:[b_1]$  & $P_2:[b_2]$ &  $P_3:[b_3]$ & $P_4:[b_5]$ & $P_5:[b_{12}] $ &
$P_6:[b_{9}]$  & $P_7:[b_{13}]$  \\ \hline
$P_1:[b_1]$  & $P_2:[b_2]$ &  $P_3:[b_4]$ & $P_4:[b_6]$ & $P_5:[b_{16}] $ &
$P_6:[b_{7}]$  & $P_7:[b_{10}]$  \\ \hline
$P_1:[b_1]$  & $P_2:[b_2]$ &  $P_3:[b_3]$ & $P_4:[b_5]$ & $P_5:[b_{12}] $ &
$P_6:[b_{8}]$  & $P_7:[b_{11}]$  \\ \hline
$P_1:[b_1]$  & $P_2:[b_2]$ &  $P_3:[b_4]$ & $P_4:[b_6]$ & $P_5:[b_{15}] $ &
$P_6:[b_{9}]$  & $P_7:[b_{14}]$  \\ \hline
$P_1:[b_1]$  & $P_2:[b_2]$ &  $P_3:[b_3]$ & $P_4:[b_5]$ & $P_5:[b_{12}] $ &
$P_6:[b_{7}]$  & $P_7:[b_{10}]$  \\ \hline
$P_1:[b_1]$  & $P_2:[b_2]$ &  $P_3:[b_4]$ & $P_4:[b_6]$ & $P_5:[b_{16}] $ &
$P_6:[b_{8}]$  & $P_7:[b_{11}]$  \\ \hline
$P_1:[b_1]$  & $P_2:[b_2]$ &  $P_3:[b_3]$ & $P_4:[b_5]$ & $P_5:[b_{12}] $ &
$P_6:[b_{9}]$  & $P_7:[b_{13}]$  \\ \hline
$P_1:[b_1]$  & $P_2:[b_2]$ &  $P_3:[b_4]$ & $P_4:[b_6]$ & $P_5:[b_{15}] $ &
$P_6:[b_{7}]$  & $P_7:[b_{10}]$  \\ \hline
$P_1:[b_1]$  & $P_2:[b_2]$ &  $P_3:[b_3]$ & $P_4:[b_5]$ & $P_5:[b_{12}] $ &
$P_6:[b_{8}]$  & $P_7:[b_{11}]$  \\ \hline
$P_1:[b_1]$  & $P_2:[b_2]$ &  $P_3:[b_4]$ & $P_4:[b_6]$ & $P_5:[b_{16}] $ &
$P_6:[b_{9}]$  & $P_7:[b_{14}]$  \\ \hline
\end{tabular}
\label{tab2c}
\caption{Relevant partitions and bamboos cutting frequencies for the example of Table \ref{tab3}: option b}
\end{table*}}
\normalsize
\noindent
and so on.
Correspondingly, if option (a) is considered, we have $\alpha=8$ and the the solution value is $8h(1)=8$. Besides,
if option (b) is considered, we have $\alpha=7$ and the solution value is $7r(3)h(3)=8.4$. Hence $H^{PW''}=8$.

\begin{prop}
If $\sum_{j=1}^nh(j) > > h(1)$, then a periodic solution exists with approximation ratio $\rho_2 = \frac{H^{PW''}}{H^*}$ converging to $\frac{12}{7} \approx 1.714$ \label{t1}
\end{prop} 

\begin{proof}
Consider applying  the above pinwheel algorithm with modified growths $h''(j)$ for $j=1,...,n$.
As $\sum_{j=1}^nh(j) > > h(1)$, then we can disregard both in the numerator and in the denominator of $\rho_2$
the contribution given by $b_1$ plus the bamboos $\in R_{S_2}$, $R_{S_3}$  and $R_{S_4}$
(as it increases numerator and denominator by at most a constant value).
Hence, for the denominator, a lower bound is given by $\frac{2}{3}(\pi_1-1) + \frac{2}{3}\pi_3+\frac{3}{4}\pi_4+\frac{\pi_2}{2}$.
Besides, the numerator is given by 
$\pi_1-1+\pi_3+\pi_4 +\pi_2$ if case (a) is considered
and by $\frac{4}{3}(\pi_1-1+\pi_3+\pi_4 +\frac{\pi_2}{2})$ if case (b) is considered.
Correspondingly, an upper bound on the approximation ratio $\rho_2$ is given by

\begin{equation}
\rho_2 \leq \frac{\min\{\pi_1-1+\pi_3+\pi_4 +\pi_2,\frac{4}{3}(\pi_1-1+\pi_3+\pi_4 +\frac{\pi_2}{2})\}}{\frac{2}{3}(\pi_1-1) + \frac{2}{3}\pi_3+\frac{3}{4}\pi_4+\frac{\pi_2}{2}}
\end{equation}

that is 

\begin{equation}
\rho_2 \leq \frac{\min\{(\pi_1-1+\pi_3+\pi_4 +\pi_2),(\pi_1-1+\pi_3+\pi_4 +\pi_2)+\frac{1}{3}(\pi_1-1+\pi_3+\pi_4 -\pi_2)\}}{\frac{2}{3}(\pi_1-1) + \frac{2}{3}\pi_3+\frac{3}{4}\pi_4+\frac{\pi_2}{2}}
\end{equation}

We observe that the worst case occurs when $\pi_4=0$ as the coefficient in the denominator of $\pi_4$ is superior to that of $\pi_1-1$ or $\pi_3$.
But then, taking this observation into account and substituting $\pi_1-1+\pi_3$ with $\alpha$ and $\pi_2$ with $\beta$,
we get

\begin{equation}
\rho_2 \leq \frac{\min\{\alpha +\beta,\frac{4}{3}(\alpha +\frac{\beta}{2})\}}{\frac{2}{3}\alpha+\frac{\beta}{2}}
= \frac{\min\{\alpha +\beta,\alpha+\beta +\frac{1}{3}(\alpha -\beta)\}}{\frac{2}{3}\alpha+\frac{\beta}{2}}
\end{equation}

where $\rho_2$ is maximum for $\alpha=\beta$. Thus we get
\begin{equation}
\rho_2 \leq \frac{2\alpha}{\frac{2}{3}\alpha+\frac{\alpha}{2}}=\frac{2}{\frac{7}{6}} = \frac{12}{7}
\label{ww}
\end{equation}

 
\end{proof}

The above proposition handles a general case with $\sum_{j=1}^n h(j) > > h(1)$. If this is not the case, the following Propositions provide approximation ratios that exhaustively handle all possible other  distributions of the $h(j)$s.

The following exhaustive cases will be considered.
\begin{enumerate}
\item $|S_2| \leq \pi_1+\pi_3+\pi_4-1$ and $\pi_1+\pi_3+\pi_4-1 \geq 1$ (Proposition \ref{p2}).
\item $|S_2| > \pi_1+\pi_3+\pi_4-1$ and $|S_2| \geq 4$, $|S_2|$ even (Proposition \ref{p3}).
\item $|S_2| > \pi_1+\pi_3+\pi_4-1$ and $|S_2| \geq 3$, $|S_2|$ odd (Proposition \ref{p4}).
\item $|S_2| = 2$ and $\pi_1+\pi_3+\pi_4-1 = 1$  (Proposition \ref{p7}).
\item $|S_2| = 2$ and $\pi_1+\pi_3+\pi_4-1 =0$ (Proposition \ref{p9}).
\item $|S_2| =1$ and $\pi_1+\pi_3+\pi_4-1 =0$ (Proposition \ref{p10}).
\item $|S_2| =0$ and $\pi_1+\pi_3+\pi_4-1 =0$ (Proposition \ref{p11}).
\end{enumerate}

\begin{prop}
If $|S_2| \leq \pi_1+\pi_3+\pi_4-1$ and $\pi_1+\pi_3+\pi_4-1 \geq 1$ then
$\rho_2 = \frac{H^{PW''}}{H^*} \leq \frac{15}{8} = 1.875$.
\end{prop} \label{p2}
\begin{proof}
If $|S_2| \leq \pi_1+\pi_3+\pi_4-1$, we consider $z(a)=\pi_R(a)+\sum_{i=1}^4\pi_i(a)$ with $|S_2|=\pi_2$
and the ratio $\frac{\pi_R(a)+\sum_{i=1}^4\pi_i(a)}{\sum_{j=1}^n h(j)}=\frac{h(1)+\pi_R(a)
+(\pi_1-1)+\pi_2+\pi_3+\pi_4}{h(1)+\frac{2}{3}(\pi_1-1) + \frac{2}{3}\pi_3+\frac{3}{4}\pi_4+\frac{\pi_2}{2}+\sum_{b_j \in R_{S_3} \cup R_{S_4}} h(j)}$. Notice that for the same argument of Proposition \ref{t1}, we observe that the worst case occurs when $\pi_4=0$.
Hence, $|S_2| \leq \pi_1+\pi_3+\pi_4-1$ corresponds to $\pi_2 \leq \pi_1-1+\pi_3$, that is 
$\beta \leq \alpha$ (with $\alpha \geq 1)$.
Also, $b_1$ is
allocated in one partition, 
all bamboos $\in R_{S_4}$ can be allocated within a second partition, and all bamboos $\in R_{S_3}$ in a third distinct partition. 
Hence, $h(1)+\pi_R(a) \leq 3$.
Besides, if $\sum_{b_j \in R_{S_3} \cup R_{S_4}}h(j) \leq \frac{1}{2}$,
then all bamboos $\in R_{S_3} \cup R_{S_4}$ could be allocated within a unique partition inducing $h(1)+\pi_R(a) \leq 2$. Thus, we have 
$\frac{h(1)+\pi_R(a)}{h(1)+\sum_{b_j \in R_{S_3} \cup R_{S_4}} h(j)}\leq 2$ with the 
worst-case given with $h(1)+\pi_R(a) = 3$ and $h(1)+\sum_{b_j \in R_{S_3} \cup R_{S_4}} h(j) =\frac{3}{2}$.
Hence, the ratio becomes 
$\frac{3+(\pi_1-1)+\pi_2+\pi_4}{\frac{3}{2}+\frac{2}{3}(\pi_1-1) + \frac{2}{3}\pi_3+\frac{\pi_2}{2}}$
that is $\frac{3+\alpha+\beta}{\frac{3}{2}+\frac{2}{3}\alpha +\frac{\beta}{2}}$ and the worst case
again is reached for $\alpha=\beta$, that is for the ratio $\frac{3+2\alpha}{\frac{3}{2}+\frac{2}{3}\alpha +\frac{\alpha}{2}} = \frac{3+2\alpha}{\frac{3}{2}+\frac{7}{6}\alpha} $. This ratio is maximum for $\alpha=1$
that is with $\rho_2 = \frac{3+2}{\frac{3}{2}+\frac{7}{6}} = \frac{15}{8}=1.875$.
\end{proof}

\begin{prop}
If $|S_2| > \pi_1+\pi_3+\pi_4-1$ and $|S_2| \geq 4$, $|S_2|$ even, then
$\rho_2 = \frac{H^{PW''}}{H^*} \leq \frac{32000}{16947} \approx 1.888$. \label{p3}
\end{prop}
\begin{proof}
If $|S_2| > \pi_1+\pi_3+\pi_4-1$ and $|S_2| \geq 4$, $|S_2|$ even,
we consider 
$z(b) = \frac{h(j^*)}{0.5}
(\pi_R(b)+\sum_{i=1}^4\pi_i(b))$ with $|S_2|=\pi_2$ and the ratio
$\frac{\frac{h(j^*)}{0.5}
(\pi_R(b)+\sum_{i=1}^4\pi_i(b))}{\sum_{j=1}^n h(j)}\leq
\frac{h(1)+\pi_R(b)
+(\pi_1-1)+\pi_3+\pi_4+\frac{\pi_2}{2}}{h(1)+\sum_{b_j \in R_{S_3} \cup R_{S_4}} h(j)+\sum_{b_j \in {S_2}}h(j)+\frac{2}{3}(\pi_1-1)+\frac{2}{3}\pi_3+\frac{3}{4}\pi_4}$.
With a similar analysis to that of Proposition \ref{p2}, 
we get $h(1)+\pi_R(b) \leq 3$
and $\frac{h(1)+\pi_R(b)}{h(1)+\sum_{b_j \in R_{S_2} \cup R_{S_3} \cup R_{S_4}} h(j)}\leq 2$ with the 
worst-case given with $h(1)+\pi_R(b) = 3$ and $h(1)+\sum_{b_j \in R_{S_3} \cup R_{S_4}} h(j) =\frac{3}{2}$. Also,
we can deduct $\pi_4=0$. 

Hence, the approximation ratio becomes 
\begin{eqnarray}
\rho_2 \leq \frac{\frac{h(j^*)}{0.5}(3+(\pi_1-1)+\pi_3 + \frac{\pi_2}{2})}{\frac{3}{2}+\frac{2}{3}(\pi_1-1)  + \frac{2}{3}\pi_3+\sum_{b_j \in {S_2}}h(j)}
\label{eq1}
\end{eqnarray}
We consider two subcases with either $h(j^*) \leq \frac{647}{1000}$,
or $\frac{4}{3} \leq h(j^*) > \frac{647}{1000}$. In the first case,
$\sum_{b_j \in {S_2}}h(j) = \frac{\pi_2}{2}$ and 
expression \ref{eq1} becomes
\begin{eqnarray}
\rho_2 \leq \frac{\frac{649}{500}(3+(\pi_1-1)+\pi_3 + \frac{\pi_2}{2})}{\frac{3}{2}+\frac{2}{3}(\pi_1-1)  + \frac{2}{3}\pi_3+\frac{\pi_2}{2}}
\label{eq2}
\end{eqnarray}

\noindent  where the worst-case occurs with $\pi_2$ as small as possible and
$(\pi_1-1)+\pi_3$ as large as possible, that is with $\pi_2=4$ and $(\pi_1-1)+\pi_3=3$.
Correspondingly, we get $\rho_2 \leq \frac{\frac{649}{500}(3+3+2)}{\frac{3}{2}+\frac{2}{3}3+2}=
\frac{\frac{1298}{125}}{\frac{11}{2}}=\frac{2596}{1375} \approx 1.888$.

In the second case, the worst-case for the ratio $\frac{h(j^*)}{0.5}$
is $\frac{4}{3}$. Besides, as $h(j^*) > \frac{649}{1000}$, then
$\sum_{b_j \in {S_2}}h(j)  > \frac{\pi_2-1}{2} + \frac{649}{1000} =
\frac{\pi_2}{2} + \frac{149}{1000}$. Hence, expression \ref{eq1} becomes
\begin{eqnarray}
\rho_2 \leq \frac{\frac{4}{3}(3+(\pi_1-1)+\pi_3 + \frac{\pi_2}{2})}{\frac{3}{2}+\frac{2}{3}(\pi_1-1)  + \frac{2}{3}\pi_3+\frac{\pi_2}{2}+\frac{149}{1000}}.
\label{eq3}
\end{eqnarray}

\noindent
Again the worst-case occurs with $\pi_2$ as small as possible and
$(\pi_1-1)+\pi_3$ as large as possible, that is with $\pi_2=4$ and $(\pi_1-1)+\pi_3=3$.
Correspondingly, we get $\rho_2 \leq \frac{\frac{4}{3}(3+3+2)}{\frac{3}{2}+\frac{2}{3}3+2+\frac{149}{1000}}=
\frac{\frac{32}{3}}{\frac{5649}{1000}}=\frac{32000}{16947} \approx 1.888$.
\end{proof}

\begin{prop}
If $|S_2| > \pi_1+\pi_3+\pi_4-1$ and $|S_2| \geq 3$, $|S_2|$ odd, then
$\rho_2 = \frac{H^{PW''}}{H^*} \leq \frac{24}{13} \approx 1.846$. \label{p4}
\end{prop}
\begin{proof}
If $|S_2|$ is odd, then $R_{S_2}$ is composed by one bamboo $b_l$ with $0.5 < h_l \leq \frac{2}{3}$ and $\pi_2=|S_2|-1$. 
We consider the following two exhaustive subcases.
\begin{enumerate}
\item $\sum_{b_j \in R_{S_3} \cup R_{S_4}} h(j) \leq \frac{1}{4}$.
If this subcase holds, we consider $z(b) = \frac{h(j^*)}{0.5}
(\pi_R(b)+\sum_{i=1}^4\pi_i(b))$
where $b_l$ can be assigned to a partition together with $\sum {b_j \in R_{S_3} \cup R_{S_4}}$ and $b(1)$ is assigned to another partition, hence 
\begin{equation}
\frac{h(1)+\pi_R(b)}{h(1)+h(l) +\sum_{b_j \in R_{S_3} \cup R_{S_4}} h(j)}\leq 2.
\end{equation}
As previously, we can deduct $\pi_4=0$ and $\frac{h(j^*)}{0.5} \leq \frac{4}{3}$.
Thus, 
\begin{equation}
H^{PW''} \leq  \frac{4}{3}(2+(\pi_1-1)+\pi_3 + \frac{\pi_2-1}{2}).
\end{equation}

Besides, 
\begin{equation}
H^* \geq \sum_{j=1}^n h(j) \geq h(1)+h(l)+\frac{\pi_{2}-1}{2}+\frac{2}{3}(\pi_1-1)
+\frac{2}{3}\pi_3 \geq  \newline
\frac{3}{2}+\frac{\pi_{2}-1}{2}+\frac{2}{3}(\pi_1-1)
+\frac{2}{3}\pi_3.
\end{equation}

Hence, we get 
\begin{eqnarray}
\rho_2 \leq \frac{\frac{4}{3}(2+(\pi_1-1)+\pi_3 + \frac{\pi_2-1}{2})}{\frac{3}{2}+\frac{2}{3}(\pi_1-1)+ \frac{2}{3}\pi_3+\frac{\pi_{2}-1}{2}}
\label{eq4}
\end{eqnarray}
 where the worst-case occurs with $\pi_2$ as small as possible and
$(\pi_1-1)+\pi_3$ as large as possible, that is with $\pi_2=3$ and $(\pi_1-1)+\pi_3=2$.
Correspondingly, we get 
\begin{eqnarray}
\rho_2 \leq \frac{\frac{4}{3}(2+2+1)}{\frac{3}{2}+\frac{2}{3}2+1}=\frac{40}{23} \approx 1.739.
\label{eq4a}
\end{eqnarray}

\item $\sum_{b_j \in R_{S_3} \cup R_{S_4}} h(j) > \frac{1}{4}$.
If this subcase holds, we consider $z(b) = \frac{h(j^*)}{0.5}
(\pi_R(b)+\sum_{i=1}^4\pi_i(b))$
where $b_l$ can be assigned to a partition together with $\sum {b_j \in  R_{S_4}}$, $b(1)$ is assigned to another partition
and $\sum {b_j \in R_{S_3}}$ to another further partition, hence 
$\frac{h(1)+\pi_R(b)}{h(1)+h(l) +\sum_{b_j \in R_{S_3} \cup R_{S_4}} h(j)}\leq 3$.
As previously, we can deduct $\pi_4=0$ and $\frac{h(j^*)}{0.5} \leq \frac{4}{3}$.
Thus, 
\begin{equation}
H^{PW''} \leq  \frac{4}{3}(3+(\pi_1-1)+\pi_3 + \frac{\pi_2-1}{2}).
\end{equation}

Besides, 
\begin{equation*}
H^* \geq \sum_{j=1}^n h(j) \geq h(1)+h(l)+\frac{\pi_{2}-1}{2}+\frac{2}{3}(\pi_1-1)+\frac{2}{3}\pi_3 + \sum_{b_j \in R_{S_3} \cup R_{S_4}} h(j) >
\end{equation*}
\begin{equation}
\hspace*{-4.7cm}
 > \frac{3}{2}+\frac{\pi_{2}-1}{2}+\frac{2}{3}(\pi_1-1)+\frac{2}{3}\pi_3+ \frac{1}{2}.
\end{equation}

Hence, we get 
\begin{eqnarray}
\rho_2 < \frac{\frac{4}{3}(3+(\pi_1-1)+\pi_3 + \frac{\pi_2-1}{2})}{\frac{3}{2}+\frac{2}{3}(\pi_1-1)+ \frac{2}{3}\pi_3+\frac{\pi_{2}-1}{2}+ \frac{1}{2}}
\label{eq4b}
\end{eqnarray}
 where the worst-case occurs always with $\pi_2=3$ and $(\pi_1-1)+\pi_3=2$.
Correspondingly, we get 
\begin{eqnarray}
\rho_2 \leq \frac{\frac{4}{3}(3+2 +1)}{\frac{3}{2}+\frac{2}{3}2+1+\frac{1}{2}}=\frac{24}{13} \approx 1.846.
\label{eq4c}
\end{eqnarray}
\end{enumerate}
\end{proof}

\begin{prop}
If $|S_2| = 2$ and $\pi_1+\pi_3+\pi_4-1 =1$ 
then
$\rho_2 = \frac{H^{PW''}}{H^*} \leq \frac{360}{193} \approx 1.865$. 
\label{p7}
\end{prop}
\begin{proof}
If $|S_2| = 2$ and $\pi_1+\pi_3+\pi_4-1 =1$,
we  consider two subcases depending on whether 
$h(j^*) \leq \frac{11}{20}$ or $h(j^*) > \frac{11}{20}$.
If $h(j^*) \leq \frac{11}{20}$, we consider 
$z(b) = \frac{h(j^*)}{0.5}
(\pi_R(b)+\sum_{i=1}^4\pi_i(b))$
and the relevant approximation ratio becomes

\begin{eqnarray}
\rho_2 \leq \frac{\frac{11}{10}(3+(\pi_1-1)+\pi_3 + \frac{\pi_2}{2})}{\frac{3}{2}+\frac{2}{3}(\pi_1-1)  + \frac{2}{3}\pi_3+\frac{\pi_2}{2}}
= 
\frac{\frac{11}{10}(3+1+ 1)}{\frac{3}{2}+\frac{2}{3}+1} =\frac{\frac{11}{2}}{\frac{19}{6}}=\frac{33}{19} \approx 1.737
\label{eq7}
\end{eqnarray}

Alternatively, $h(j^*) > \frac{11}{20}$ and 
we consider $z(a)=\pi_R(a)+\sum_{i=1}^4\pi_i(a)$  with $|S_2|=\pi_2$
and the relevant approximation ratio becomes

\begin{eqnarray}
\rho_2\leq \frac{3+(\pi_1-1)+\pi_2+\pi_3}{\frac{3}{2}+\frac{2}{3}(\pi_1-1) + \frac{2}{3}\pi_3+\frac{\pi_2}{2}+ \frac{1}{20}}
=\frac{3+3}{\frac{3}{2}+\frac{2}{3}+1+\frac{1}{20}}=
\frac{6}{\frac{193}{60}}=\frac{360}{193} \approx 1.865.
\label{eq7b}
\end{eqnarray}
\end{proof}

\begin{prop}
If $|S_2| = 2$ and $\pi_1+\pi_3+\pi_4-1 =0$ 
then
$\rho_2 = \frac{H^{PW''}}{H^*} \leq \frac{100}{53} \approx 1.887$. \label{p9}
\end{prop}
\begin{proof}
If $|S_2| = 2$ and $\pi_1+\pi_3+\pi_4-1 =0$, let $b_k,b_l$ $\in S_2$
with $h(k) \leq h(l)$.
We consider three main subcases.
\begin{enumerate}
\item $h(k) \leq \frac{23}{40}$ and $(\sum_{b_j \in R_{S_3}} h(j) \leq \frac{3}{8}$ or $\sum_{b_j \in R_{S_4}} h(j \leq \frac{1}{3}$ or $h(l) \leq \frac{23}{40})$. If this subcase holds, then bamboo $b_k$ can be assigned to a partition together either with $\sum {b_j \in R_{S_3}}$
(if $\sum_{b_j \in R_{S_3}} h(j) \leq \frac{3}{8}$ holds) or
with $\sum {b_j \in R_{S_4}}$
(if $\sum_{b_j \in R_{S_4}} h(j) \leq \frac{1}{3}$ holds)
or with $b_l$ (if $h(l) \leq \frac{23}{40}$ holds) and the total number 
of partitions is $\leq 4$ with $H^{PW''} \leq \frac{23 \times 4}{20}=\frac{23}{5}$. On the other hand, we have 
$H^* \geq \frac{3}{2}+\frac{2}{2}=\frac{5}{2}$.
Correspondingly, we get 
$\rho_2 = \frac{H^{PW''}}{H^*} \leq \frac{\frac{23}{5}}{\frac{5}{2}}=\frac{46}{25} = 1.84$. 
\item 
$h(k) \leq \frac{23}{40}$ and $\sum_{b_j \in R_{S_3}} h(j) > \frac{3}{8}$ and $\sum_{b_j \in R_{S_4}} h(j) > \frac{1}{3}$ and
$\frac{2}{3} > h(l) > \frac{23}{40}$.
If this subcase holds, then we consider
the ratio $\rho_2\leq\frac{\pi_R(a)+\sum_{i=1}^4\pi_i(a)}{\sum_{j=1}^n h(j)}$ 
where $\pi_R(a)+\sum_{i=1}^4\pi_i(a) \leq 5$ as $5$ partitions are necessary at most to allocate $b(1)$, $b(k)$, $b(l)$, $\sum {b_j \in R_{S_3}}$ and $\sum {b_j \in R_{S_4}}$. Also, $\sum_{j=1}^n h(j) =
h(1)+h(k)+h(l)+\sum_{b_j \in R_{S_3}} h(j) + \sum_{b_j \in R_{S_4}} h(j) \geq 1+\frac{1}{2}+\frac{23}{40}+\frac{3}{8}+\frac{1}{3}$. 
Correspondingly, we get 
$\rho_2\leq\frac{5}{1+\frac{1}{2}+\frac{23}{40}+\frac{3}{8}+\frac{1}{3}}=
\frac{300}{167} \approx 1.796$.
\item $h(k) > \frac{23}{40}$.
If this subcase holds, then we consider
the ratio $\rho_2\leq\frac{\pi_R(a)+\sum_{i=1}^4\pi_i(a)}{\sum_{j=1}^n h(j)}$ 
where $\pi_R(a)+\sum_{i=1}^4\pi_i(a) \leq 5$. For the denominator,
we get $\sum_{j=1}^n h(j) =
h(1)+h(k)+h(l)+\sum_{b_j \in R_{S_3}} h(j) + \sum_{b_j \in R_{S_4}} h(j) \geq 1+\frac{23}{20}+\frac{1}{2}= \frac{53}{20}$ where we assume $\sum_{b_j \in R_{S_3} \cup R_{S_4}} h(j) \geq \frac{1}{2}$
or else four partitions only would be necessary for the numerator and a better approximation ratio would hold.
Correspondingly, we get 
$\rho_2\leq\frac{5}{\frac{53}{20}}=
\frac{100}{53} \approx 1.887$.
\end{enumerate}
\end{proof}
 
\begin{prop}
If $|S_2| = 1$ and $\pi_1+\pi_3+\pi_4-1 =0$
then
$\rho_2 = \frac{H^{PW''}}{H^*} \leq \frac{20}{11} \approx 1.818$. \label{p10}
\end{prop}
\begin{proof}
If $|S_2| = 1$ and $\pi_1+\pi_3+\pi_4-1 =0$,
then subset $S_1$ is composed by bamboo $b_1$ only, $\sum_{j \in S_3}h(j) < \frac{1}{2}$, $\sum_{j \in S_4}h(j) < \frac{2}{3}$
and $\sum_{j \in S_3 \cup S_4}h(j) < \frac{7}{6}$.
Also, $S_2$ is composed by a single bamboo $b_l$ with grow $\frac{1}{2} < h(l) \leq \frac{2}{3}$. Then, if $\sum_{j \in S_3 \cup S_4}h(j) \leq \frac{1}{2}$, all bamboos $\in  S_3 \cup S_4$ can be allocated to a 
single partition and thus $H^{PW''}=3$ and 
$\rho_2 = \frac{H^{PW''}}{H^*} \leq \frac{3}{2}$ as $H^* \geq 2$ always holds.
Alternatively, $\sum_{j \in S_3 \cup S_4}h(j) > \frac{1}{2}$
and correspondingly $\sum_{j=1}^n h(j) > 2$. Then, if $h(l)+\sum_{j \in S_3 \cup S_4}h(j) \geq \frac{6}{5}$,
$H^*\geq  \sum_{j=1}^n h(j) \geq \frac{11}{5}$: by allocating $b_1$ to one partition,
$b_l$ to another partition, all bamboos $\in  S_3$ to a third partition
and all bamboos $\in  S_4$ to a fourth partition, we get
$H^{PW''}=4$ and 
$\rho_2 = \frac{H^{PW''}}{H^*} \leq \frac{4}{\frac{11}{5}} = \frac{20}{11}$. Finally, if $ 2 < \sum_{j=1}^n h(j) < \frac{11}{5}$,
where $h(l) > \frac{1}{2}$ and $\sum_{j \in S_3 \cup S_4}h(j) > \frac{1}{2}$, we have both $h(l) < \frac{3}{5}$, and $\sum_{j \in S_3 \cup S_4}h(j) < \frac{3}{5}$. But then,
either  $\sum_{j \in S_3} h(j) < \frac{3}{10}$ or
$\sum_{j \in S_4} h(j) < \frac{3}{10}$ necessarily holds.
In both cases, by allocating either all bamboos $\in  S_3 $
or all bamboos $S_4$ in the same partition of $b_l$, we 
need $3$ partitions only and, as $h(l) < \frac{3}{5}$, we have
$H^{PW''} \leq \frac{6}{5} 3 = \frac{18}{5}$.
As $H^* \geq  \sum_{j=1}^n h(j) > 2$, we get
$\rho_2 = \frac{H^{PW''}}{H^*} \leq \frac{\frac{18}{5}}{2}
= \frac{9}{5}$.
\end{proof}

\begin{prop}
If $|S_2| =0$ and $\pi_1+\pi_3+\pi_4-1 =0$
then
$\rho_2 = \frac{H^{PW''}}{H^*} \leq \frac{3}{2}=1.5$. \label{p11}
\end{prop}
\begin{proof}
If $|S_2| =0$ and $\pi_1+\pi_3+\pi_4-1 =0$,  
then subset $S_1$ is composed by bamboo $b_1$ only, $\sum_{j \in S_3}h(j) < \frac{1}{2}$ and $\sum_{j \in S_4}h(j) < \frac{2}{3}$.
Also, $S_2$ is empty and, correspondingly,
$\sum_{j=1}^n h(j) < \frac{13}{6}$ and
$\sum_{j=2}^n h(j) < \frac{7}{6}$ as $h(1)=1$. But then,
by allocating all bamboos $\in S_3$ in one partition and all bamboos $\in S_4$ in another partition, we get 
$H^{PW''} = 3$, Correspondingly, as $H^* \geq 2$,
$\rho_2 = \frac{H^{PW''}}{H^*} \leq \frac{3}{2}$.
\end{proof}

\begin{coro}
The approximation ratio of Algorithm $A_2$ is not superior to 
$\frac{32000}{16947} \approx 1.888$.
\end{coro}
\begin{proof}
Putting together Propositions \ref{p2}, \ref{p3}, \ref{p4}, \ref{p7}, \ref{p9}, \ref{p10} and \ref{p11}, the Corollary immediately holds, the worst case occurring on Proposition \ref{p3}.
\end{proof}

\section{Conclusions}
An improved pinwheel algorithm has been proposed for the $BGT$ problem reaching an approximation ratio $\rho_2$ that in the worst-case converges to $\frac{12}{7}$ for instances with $\sum_{i=1}^nh(i) > > h(1)$. 
Also, it is shown that the worst case on the other instances
is reached when $\sum_{i=1}^nh(i) < 6 h(1)$ with $\rho_2 \leq \frac{32000}{16947} \approx 1.888$. 
This  approximation ratio is
substantially due to the presence of bamboos $b_j \in S_2$ with
heights $\frac{1}{2}h(1) < h(j) \leq \frac{2}{3}h(1)$ (inequality (\ref{ww}) induces $\rho_2 \leq \frac{3}{2}$ when $|S_2|=\beta=0$) or to specific subcases
where $\sum_{i=1}^nh(i)$ is relatively small ($\sum_{i=1}^nh(i)< 6h(1)$).
Hence, a research direction worthy of investigation would be to find a way to compute an improved lower bound with value strictly greater than $LB = \max \{2h(1);\sum_{j=1}^n h(j)\}$ for these subcases.

\section*{Acknowledgement}
This work was funded by the GEO-SAFE project and the EU Horizon2020 RISE programme,  grant agreement No 691161.

\end{document}